\documentclass[sigconf]{acmart}
 
\usepackage{booktabs} 
\usepackage{balance}
\usepackage{graphicx}
\usepackage{subcaption}
\usepackage{amsmath}
\usepackage{amsthm}

\usepackage{algorithm}
\usepackage[noend]{algorithmic}
\usepackage{pifont}

\usepackage{multirow}
\usepackage{capt-of}
\usepackage{varwidth}
\usepackage[switch]{lineno} 
\DeclareMathAlphabet\mathbfcal{OMS}{cmsy}{b}{n}

\DeclareMathOperator*{\argmin}{arg\,min}

\renewcommand\footnotetextcopyrightpermission[1]{} 

\AtBeginDocument{%
  \providecommand\BibTeX{{%
    \normalfont B\kern-0.5em{\scshape i\kern-0.25em b}\kern-0.8em\TeX}}}

\begin{document}
\title{CARAMEL: A Succinct Read-Only Lookup Table via Compressed Static Functions}

\author{Benjamin Coleman}
\affiliation{%
  \institution{Google DeepMind}
  \city{Mountain View}
  \country{USA}
}
\email{colemanben@google.com}

\author{David Torres Ramos}
\affiliation{%
  \institution{Distill}
  \city{New York}
  \country{USA}
}
\email{detorresramos1@gmail.com}

\author{Vihan Lakshman}
\affiliation{%
  \institution{MIT CSAIL}
  \city{Cambridge}
  \country{USA}
}
\email{vihan@mit.edu}

\author{Chen Luo}
\affiliation{%
  \institution{Amazon}
  \city{Palo Alto}
  \country{USA}
}
\email{cheluo@amazon.com}

\author{Anshumali Shrivastava}
\affiliation{%
  \institution{Rice University}
  \city{Houston}
  \country{USA}
}
\email{anshumali@rice.edu}

\begin{abstract}
Lookup tables are a fundamental structure in many data processing and systems applications. Examples include tokenized text in NLP, quantized embedding collections in recommendation systems, integer sketches for streaming data, and hash-based string representations in genomics. With the increasing size of web-scale data, such applications often require compression techniques that support fast random $O(1)$ lookup of individual parameters directly on the compressed data (i.e. without blockwise decompression in RAM). While the community has proposd a number of succinct data structures that support queries over compressed representations, these approaches do not fully leverage the low-entropy structure prevalent in real-world workloads to reduce space. Inspired by recent advances in static function construction techniques, we propose a space-efficient representation of immutable key-value data, called CARAMEL, specifically designed for the case where the values are multi-sets. By carefully combining multiple compressed static functions, CARAMEL occupies space proportional to the data entropy with low memory overheads and minimal lookup costs. We demonstrate 1.25-16x compression on practical lookup tasks drawn from real-world systems, improving upon established techniques, including a production-grade read-only database widely used for development within Amazon.com. 
\end{abstract}




\maketitle
\pagestyle{plain} 

\section{Introduction}

Efficient data structures for storing key-value pairs are indispensable components for countless modern systems, spanning a diverse array of applications including recommender systems, search engines, genomics, and natural language processing. With the ballooning volume of web-scale data, practitioners maintain an intense interest in building caches with a lower memory footprint and faster query times. Inspired by the observation that many applications involve static data collections that need not support real-time insertions or deletions, which allows for more aggressive compression strategies, we focus on developing a read-only data structure with space close to the information-theoretic lower bound while supporting constant-time lookups directly on the compressed representation. 

In this work, we design such a representation using \textit{compressed static functions} and call our data structure CARAMEL (Compressed Array Representation And More Efficient Lookups). A \emph{static function} is a mapping $f: K \to V$ from a static, known set of \emph{hashable} keys $K$ to \emph{encodable} values $V$ such as strings and integers. A \emph{compressed static function} (CSF) is a static function that can represent $f$ in space proportional to the \textit{entropy} of $V$ and can evaluate $f$ in $O(1)$ time~\cite{dietzfelbinger2008succinct}.
While CSFs were historically a purely theoretical curiosity, recent advances in practical construction techniques allow near-optimal sized CSFs to be developed for millions of parameters~\cite{8416577,genuzio2016fast,genuzio2020fast}. 

However, prior work on utilizing CSFs in practical applications largely focus on storing a single value for a given key \cite{genuzio2016fast, shibuya2022space}. In this paper, we study the problem of how to efficiently build a CSF structure for a multiset of values. This setting arises in a number of real-world application. For example, search and recommendation systems often cache the most popular query-result pairs to reduce the throughput demands on machine learning models and improve the average latency~\cite{Luo2022}. In large-scale production serving systems, this table can contain millions of keys each with hundreds of values and consumes the majority of the server memory budget~\cite{guo2020accelerating,nigam2019semantic}. Furthermore, common NLP and genomics tokenization processes yield massive arrays of integer tokens~\cite{shu2018compressing,chen2018learning,ondov2016mash}. To that end, we propose a novel data structure based on arrays of multiple CSFs to address these data-intensive applications. By carefully designing our data structure to support multiple CSFs with minimal overhead, we find that we can achieve further compression than simply using a single static function and validate our approach across a myriad of real-world static retrieval tasks.  

With this goal in mind, we consider the following problem. Given a static set of pairs consisting of a key $k$ and a list of values $A = \left\{(k_i, [v_{i1}, v_{i2}, \dots, v_{im}])\right\}_{i=1}^n$, we wish to construct a small data structure $S$ that can return an individual value $v_{ij}$ in constant time and hence the entire list in time $O(m)$. Standard file compression methods that employ methods such as run-length encoding, delta coding, specialized codecs and other techniques~\cite{witten1999managing,scholer2002compression,williams1999compressing,thompson2022github}are insufficient for this task because they do not allow queries directly on the compressed representation. Moreover, existing succinct data structures do not fully leverage the underlying structure of the data to achieve lower space. In summary, our ideal data representation should satisfy three requirements: \\
\textbf{1) Lossless Compression:} $S$ must represent $A$ with no error, and $S$ should be much smaller than $A$. Ideally, $S$ is of similar size to that obtained via standard compression methods. \\
\textbf{2) Random Access:} The cost to obtain a value from $S$ should scale $O(1)$ with $n$ and be fast in practice. Ideally, the lookup time should be comparable to the latency of a hash table. \\
\textbf{3) Fast Construction:} The time and computation needed to construct $S$ should not exceed practical limitations. For example, a standard workstation should be able to compress an input with a million keys in a few minutes.

\subsection{Motivating Applications}

\textbf{Machine Learning Prediction Serving:} Compressed array representations can serve as space-efficient databases for serving precomputed predictions from machine learning models. In large-scale online settings, such as in search engines, latency and throughput limitations make it infeasible to process every query through a large, expensive ML model. Therefore, a common optimization is to cache the outputs for frequent queries \cite{Luo2022}. This mapping must have $O(1)$ access time and be space-efficient to handle the data volume of modern systems. Moreover, the reduced space footprint enabled by a compressed representation can be especially beneficial for storing and serving results in edge computing environments. 

\textbf{Tokenization:} Tokenization, the process of breaking down a text into smaller units, is an essential step for nearly all NLP modeling tasks. To feed this tokenized text into an NLP model such as BERT, practitioners typically map each token to a unique integer ID \cite{devlin2018bert, song2020linear, kudo2018sentencepiece}. For very large text corpora, it quickly becomes infeasible to store the entire tokenized dataset in main memory for model training. Consequently, practitioners have adopted more sophicated data pipelines that tokenize text on the fly in parallel to training on a separate batch of data. This approach, though, adds additional engineering complexity to a model training framework and may even become a computational bottleneck when training large but shallow models, such as those commonly found in recommender systems \cite{naumov2019deep}. In addition, this approach hinders the ability to access arbitrary training examples, preventing the use of recently developed algorithms for dynamic negative sampling, importance sampling and other batch selection approaches~\cite{daghaghi2021tale,coleman2019selection,katharopoulos2018not,jiang2019accelerating}. These challenges motivate the need for a compressed representation of the tokenized data, which is simply a matrix of integers and thus a candidate for CSF compression. 

\textbf{Genomics:} Given the massive scale of modern genome datasets, genomics has emerged as a particularly critical domain for compressed representations that allow for efficient lookups. Much like text processing, genome datasets involve representing sequences of nucleotides as strings. Modern genomics archives contain petabytes of data, and the data generation rate has only increased with new developments in sequencing technology~\cite{elworth2020petabytes}. Applications such as genome search~\cite{gupta2021fast}, $k$-mer coverage estimation~\cite{brown2012reference}, metagenomic analysis~\cite{segata2012metagenomic} and several others all require mappings from $k$-mer strings (genomic $n$-grams) to integers. Depending on the application, these integer IDs can refer to sequencing coverage metrics, species identifiers, minhash-based featurizations~\cite{ondov2016mash} and other meta-information. The genomics community has recently studied the application of CSFs prior to our work for the purpose of mapping sequence reads to coverage numbers during assembly~\cite{shibuya2022space}. 

In this work, we apply CARAMEL to compress RefSeq, a collection of reference sequences for genomic DNA, transcripts and proteins~\cite{pruitt2006ncbi}. Compressed representations of large reference sequence databases are incredibly important for the genomics community because they allow biological analyses to be conducted against larger references and with fewer computational resources. It is for this reason that the Mash featurization, a way of representing sequences as integer hash codes, became a highly popular first step in many analysis pipelines. We focus on compressing this featurization even further, which could enable analysis to scale to more problems and computational environments.


\subsection{Contributions}

We observe that several tasks in data-intensive systems -- quantized embedding lookups, precomputed result caches, genomics, and more -- critically rely on fast key-value lookups over a multiset of values. Based on recent advances in static function representations, we propose a compressed data structure for fast lookups while achieving a near-optimal space footprint. Overall, our contributions can be summarized as follows: \\

\textbf{Compressed Memory for Parameter Arrays:} We extend compressed static functions to build a novel data structure consisting of arrays of CSFs, which we call CARAMEL. We show that our structure provides $O(m)$ access to rows, memory proportional to the column-wise entropy of the matrix, and removes the need to explicitly store the keys. We also provide evidence that an array of CSFs provides more potential for compression via implicit parameter groupings. \\

\textbf{Theoretical and Practical Improvements:} We improve the CSF construction process by introducing a greedy heuristic to permute the rows of the input matrix to minimize the column entropies, achieving up to 40\% additional compression for settings where the order of tokens is unimportant.  \\

\textbf{Practical Applications:} To our knowledge, we provide the first demonstration that CSFs can be applied to practical applications in a variety of data systems, such as text tokenization, quantized embedding tables, count sketches, genome datasets, and caching search results. In these settings, we show that CSFs can achieve a significant reduction in memory footprint while maintaining acceptable construction times and lookup latencies (Table~\ref{tab:csf_space}). In fact, we find that CARAMEL can achieve faster query times than a standard hash table (see Figure~\ref{fig:ablations}), which underscores the practical relevance of our proposal. We also conduct a case study comparing CARAMEL to a production grade read-only database used widely within Amazon.com, where we find that our proposed method achieves significantly improved compression with comparable construction times. Finally, we open source the Python library\footnote{https://github.com/brc7/caramel} we developed for conducting experiments which, to the best of our knowledge, is the first publicly available CSF implementation to operate on general hashable types as opposed to only integers. We expect this implementation to be a useful reference for the community, especially in genomics where the lack of a Python module is a serious barrier to the adoption and use of CSF structures~\cite{shibuya2022space}.

\section{Background}
\label{sec:background}
In this section, we introduce the tools needed for our proposed data structure, beginning with compressed static functions. A CSF, first formally introduced by \cite{belazzougui2013compressed}, is a function satisfying the following definition. Constructing such functions is the subject of considerable research in the theory community~\cite{porat2009optimal,dietzfelbinger2008succinct,hreinsson2009storing,belazzougui2013compressed,botelho2013practical}.

\begin{definition}
\label{def:csf_function}
Given a set $S$ of key-value pairs $S = \{(k,v(k))\}$ and a universe $U$ of possible keys, let $K \subseteq U$ denote the set of $N = |S|$ unique keys and $V$ denote the multiset of values $v(k) \in V$. A compressed static function $f_S(k) \to v$ is a function which returns $v(k)$ if $k \in K$ and any value if $k \not \in K$.
\end{definition}

It should be noted that the values $V$ are a \textit{multiset} because the elements are not necessarily unique. For example, we may require that $f(k_1) = f(k_2)$. When combined with the fact that CSFs are \textit{static}, the non-uniqueness of values will permit us to attain good compression by developing an optimal instantaneous code for the values of $V$. To quantify the space needed by the CSF, we introduce the concept of empirical entropy.

\begin{definition} Given a multiset $V$ of $N$ values, let $\mathrm{supp}(V)$ be the set of unique values, or support, of $V$ and $\#(v)$ be the number of occurrences of $v \in V$. The first-order entropy of $V$ is:
$$ H_0(V) = \sum_{v \in \mathrm{supp}(V)} \frac{1}{\#(v)} \log_2 \frac{N}{\#(v)}$$
\end{definition}

Compressed static functions aim to represent $f(k)$ in $O(N H_0(V))$ space \cite{belazzougui2013compressed}. Several different construction methods are possible. The simplest method is to use an $O(N)$-space bijection from $K$ to the integers $\{1, ... N\}$ to index into an array of encoded values (e.g. using a minimal perfect hash function and Huffman coding). However, this technique requires substantial overhead in the form of the bijection and pointers to the array. Recent work demonstrates substantial improvements by \textit{directly computing} the encodings of the values.

\textbf{CSFs via Linear Systems:} The fundamental insight of \cite{8416577} is that a CSF may be constructed via a binary linear system. To illustrate this construction technique, suppose we wish to construct a CSF for the set of key-value pairs $\{(k_1, v_1), (k_2, v_2), ... (k_N, v_N)\}$. We begin by hashing each key using three universal (random) hash functions, resulting in three integers $(h_1(k), h_2(k), h_3(k))$ for each key. We also apply an instantaneous code to obtain a binary encoding $e(v)$ for each value. This process gives rise to a linear system that, when solved, yields a CSF.

For a concrete example, let's write out a portion of the linear system for the keys $k_1, k_2,$ and $k_N$. Suppose that the hash functions output values in the range $[1, 8]$ and that $k_1$ maps to $(5, 7, 8)$, $k_2$ to $(1, 4, 2)$ and $k_N$ to $(3, 6, 8)$. Furthermore, suppose that we have the encodings $\{e(v_1) = 1, e(v_2) = 001, ... e(v_N) = 1\}$. We consider the following linear system on the binary space $\mathbf{Z}/2\mathbf{Z}$, where addition and multiplication are replaced by the OR and XOR operations, respectively.

$$ \begin{bmatrix}
0 & 0 & 0 & 0 & 1 & 0 & 1 & 1\\
1 & 1 & 0 & 1 & 0 & 0 & 0 & 0\\
0 & 1 & 1 & 0 & 1 & 0 & 0 & 0\\
0 & 0 & 1 & 1 & 0 & 1 & 0 & 0\\
 &  &  & \vdots & & & & \\ 
0 & 0 & 1 & 0 & 0 & 1 & 0 & 1\\
\end{bmatrix}
\cdot
\mathbf{g}
=
\begin{bmatrix}
e(v_1)_1 \\
e(v_2)_1 \\ 
e(v_2)_2 \\ 
e(v_2)_3 \\ 
\vdots \\ 
e(v_N)_1
\end{bmatrix}
=
\begin{bmatrix}
1 \\
0 \\ 
0 \\ 
1 \\ 
\vdots \\ 
1 
\end{bmatrix}$$
In the solution vector, we use the notation $e(v_j)_i$ to denote the $i$th bit of the instaneous encoding for the value $v_j$. The positions of the 1's in the matrix are decided by the hash functions of the corresponding key $(h_1(k), h_2(k), h_3(k))$, offset by the position of the corresponding bit of the encoding $e(v(k))$. For example, the first row in the matrix corresponds to $k_1$, as $v_1$ has a 1-bit encoding. The next three rows are responsible for the 3-bit encoding of $v_2$. Thus, there is one linear equation for each bit of each encoded value. The solution $\mathbf{g}$ of this system, if it exists, can serve as a lookup table to compute $f(k)$. To obtain the first bit of $e(f(k))$, we access the bits of $\mathbf{g}$ at the locations identified by the hash values of $k$ and compute the XOR. To obtain the subsequent bits of the encoding, we shift the locations to the right. Because $e$ is an instantaneous code, the process finishes as soon as we read a valid symbol.

The probability that the system is solvable depends on the length of $\mathbf{g}$ and is governed by the satisfiability of $s$-XORSAT instances (where $s$ is the number of hash functions used)~\cite{dubois20023}. The major contribution of \cite{8416577,genuzio2016fast,genuzio2020fast} is to make these large systems practically solvable.

\begin{theorem} \label{thm:csf_size} Given the values in Definition~\ref{def:csf_function}, let $e$ be a prefix-free encoding for $V$ where $v \in V$ is represented by code word $e(v)$ with length $l_e(v)$. Then the static function of~\cite{hreinsson2009storing,8416577} occupies space:
$$ \delta_s \sum_{v \in V} l_e(v) + O(1)$$
where $\delta_s$ is a constant and the $O(1)$ term is the cost to store the codebook for the encoding $e$.
\end{theorem}

\textbf{Sources of Overhead:} When the code lengths from Theorem~\ref{thm:csf_size} exactly match the inverse frequencies of the values, the sum of code lengths becomes exactly $N H_0(V)$, leading to an overhead of $\delta_s$ (roughly 8.9\% for $s=3$). However, this is not the case in practice - we observe a much larger per-key overhead. There are two reasons for this: code length mismatch and codebook storage. While Huffman codes have an expected per-key overhead within 1 bit of the Shannon bound ($H_0$), the encoding only hits the bound with equality when frequencies decay in inverse powers of two. Also, the codebook requires storage of $V$, which can be expensive when $V$ consists of long strings or integers.

\section{Related Work}
\label{sec:relatedwork}
In this section, we introduce two algorithms that can be used for similar purposes as CARAMEL. We compare against both algorithms in the experiments.


\subsection{Minimal Perfect Hash Tables}
The problem of minimal perfect hashing is closely related to CSFs. A minimal perfect hash function is a function $f(k)$ that maps a set of unique keys $K = \{k_1, ... k_N\}$ to a permutation of the consecutive integers $\{1, ... N\}$. Clearly, minimal perfect hash (MPH) functions are special cases of static functions (where all the values are different). Similar construction methods (e.g. hypergraph peeling~\cite{belazzougui2014cache}) can be used to obtain both CSFs and MPH functions, and similar bounds on the size (in terms of the number of bits / key) exist for MPH functions. 

MPH functions provide a straightforward baseline for the compression problems we consider here. We begin by constructing a MPH from the set of keys to a set of memory locations. Then, we place a compressed representation of the data or value at the corresponding location. This approach, described in further detail by~\cite{belazzougui2013compressed}, has larger overhead in practice than a CSF-based representation but can still provide reliable performance. In our experiments we compare against such an implementation, which was written in Java and is used in production by Indeed.

\subsection{Succinct Data Structures}

There is a large body of literature on succinct data structures, which aims to develop compressed implementations of data types that support certain operations in $O(1)$ time. The CSF representation that we use was developed as a succinct approach to the dictionary lookup problem, which is to implement a compressed key-value store with $O(1)$ keyed lookup queries. However, many other types of succinct data structures have been developed, and some of them have been used for practical applications of compression~\cite{agarwal2015succinct,duan2021succinct}. 

One particularly useful class of methods focus on succinct string data structure, which support byte-level acccess on the compressed representation. Techniques such as wavelet trees, Burrows-Wheeler transforms and compressed tries can be used to implement such a data structure, and there are several options available. We compare against the succinct data structure from~\cite{agarwal2015succinct}, which is designed to provide fast random access to byte strings using a variation on the Burrows-Wheeler transform.

\section{Algorithm}

In this section, we describe our approach to compress key-value datasets through (CARAMEL). For clarity of the mathematical presentation, we will describe our algorithm in terms of partitioning an integer matrix $A\in \mathbb{Z}^{N\times m}$ into $m$ different CSFs, one for each dimension. We note that CARAMEL can also operate on general hashable types such as strings. However, we can discuss the theoretical properties of CARAMEL in terms of integer matrices without loss of generality by assuming that the integers are the hashed representations of some initial input (which is also how we implement the data structure in practice). We also consider the set of keys to be the integers $\{1, ... N\}$. In practice, for applications such as embedding lookup, we use the text string corresponding to each row of $A$ as the key but we can again consider integer keys in our analysis without loss of generality.


\begin{figure}[t]
\begin{center}
\centerline{\includegraphics[width=2.2in]{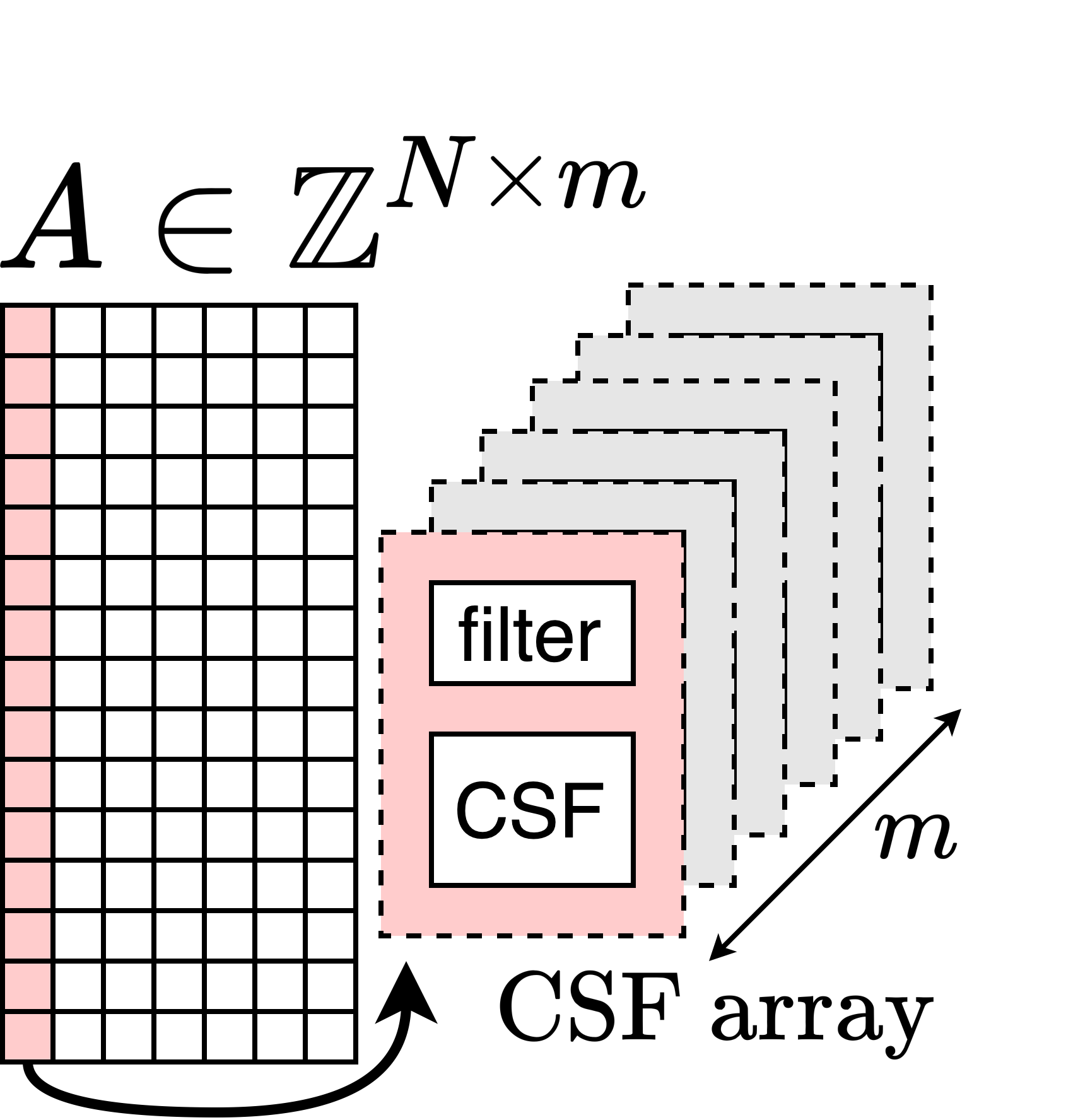}}
\end{center}
\end{figure}

\textbf{Construction and Query:} The construction process consists of the following tasks, which are described by Algorithm~\ref{alg:construction}. First, if the application supports the shuffling of items without changing the representation, we permute the items in $A$ to find the configuration with the lowest column-wise entropy (see Definition~\ref{def:permute_problem}). This improves both the construction time and space of the CSF.

\textbf{Discarding the Keys:} The only practical constraint on the keys is that they be \textit{hashable}. This allows integers, strings and many other types to be used. It should be noted that each CSF does not require the keys to be maintained or stored in any way. This property of the CSF structure allows us to achieve additional memory savings over alternative methods.

\textbf{Why not use a single CSF?} Since we may just as easily treat the entire matrix $A$ as a flat set of integer parameters, it is worth questioning whether a single CSF might result in a more efficient representation. However, there are good reasons for using a 2D representation. For example, there are information-theoretic benefits to independently encoding each dimension when different columns have different vocabulary distributions. For example, suppose we wish to compress a matrix with two columns. Column 1 has the vocabulary $(a, b, c)$ with relative frequencies $(0.5, 0.25, 0.25)$ while column 2 has the vocabulary $(d, e, f)$ with the same frequencies. If we compress the columns independently, $a$ and $d$ are both representable by 1 bit. However, the Huffman code for the merged version uses 2 bits for both $a$ and $d$ (while having the same codebook storage cost). Because CARAMEL is able to implicitly encode the knowledge that different columns have different vocabulary distributions, we can use specialized encodings within each column, enabling more optimizations (e.g. Algorithm ~\ref{alg:permute_greedy}).

\textbf{Exploiting Parallelism:} There are also practical reasons to decompose a large set of integer parameters into several independent CSFs. We find it advantageous to parallelize the construction process, since our larger tasks (e.g. MSMarco~\cite{nguyen2016ms} and the Pile~\cite{gao2020pile}) would need systems with billions of equations to represent the flat version. The CSF array representation is also helpful at query time, where we may query multiple structures in parallel.

\section{Theory}
\label{sec:theory}

In this section, we discuss the space and query time of our CSF array. In particular, we discuss optimizations to reduce the column-wise entropy of the matrix for applications that support the reordering of items along one of the dimensions of $A$.

\textbf{Memory and Lookup Time:} We begin with an easy consequence of the results in~\cite{hreinsson2009storing}. The space occupied by our data structure is proportional to the sum of column entropies on $A$.

\begin{corollary} Given a matrix $A \in \mathbb{Z}^{N\times m}$, let $A_i$ be the multiset of values in column $i$ (i.e. $A_i = \{A_i\}_{j=1}^N$). Our representation occupies space $O\left(N\sum_i H_0(A_i)\right)$ and has $O(1)$ query time.
\end{corollary}

\begin{algorithm}[t]
\caption{Permute Matrix}
\label{alg:permute_greedy}
\begin{algorithmic}
  \STATE {\bfseries Input:} Minimum block size $B$, $A\in \mathbb{Z}^{N\times m}$ without duplicates in each row
  \STATE {\bfseries Output:} Permuted $A \in \mathbb{Z}^{N \times m}$
  \STATE Row locations $T = $ dictionary from value $v$ to rows $(r_1, r_2, ... )$ of $A$ containing $v$
  \STATE Ineligible rows $E = $ dictionary from column $c$ to $(\,)$
  \WHILE {$\min_c |E(c)| < m$ (there are eligible rows remaining)}
  \STATE $V = $ values $(v_1, v_2, ... )$ in descending order of $|T(v)|$
  \STATE $P = (\,)$ of (value, destination column, destination weight)
  \FOR {$v \in V$}
    \STATE Best column $c' = \mathrm{argmax}_c |T(v)\setminus E(c)|$ with best weight $w' = |T(v)\setminus E(c')|$
    \STATE Add $(v, c', w')$ to $P$
  \ENDFOR
  \STATE $v^{\star}, c^{\star}, w^{\star} = $ entry of $P$ with the highest weight $w$
  \FOR {eligible row $r \in T(v^{\star})\setminus E(c^{\star})$}
    \STATE Find $c$ such that $A_{r,c} = v^{\star}$
    \STATE Swap $A_{r,c}$ and $A_{r,c^{\star}}$
    \STATE Remove $r$ from $T(v^{\star})$ and add $r$ to $E(c^{\star})$
  \ENDFOR
  \IF {$w^{\star} \leq B$}
  \STATE Early terminate and return $A$
  \ENDIF
  \ENDWHILE
\end{algorithmic}

\end{algorithm}

\subsection{Row Permutation to Minimize Column Entropy}

There are several applications for which the order of integers in each row of the matrix does not affect the validity of the representation. Such a situation can arise in product recommendation, where each row is a set of pre-computed product IDs, or in genomics, where each row is a set of minhash values. In these situations, we can further reduce the space by finding the order of integers that minimizes $H_0(V)$ for each column of the matrix. We wish to perform the following minimization:
 
\begin{definition} \label{def:permute_problem} Given a matrix $A \in \mathbb{Z}^{N\times m}$, let $\mathcal{A}(A)$ be the set of matrices such that for all $A' \in \mathcal{A}(A)$ and indices $(i,j)$, $A'_{(i,j)} = A_{(i, \pi_i(j))}$, where $\pi_i$ is a permutation function. For notational convenience, let $A_i$ be the multiset of values in column $i$ (i.e. $A_i = \{A_i\}_{j=1}^N$). The \emph{row permutation problem} is to perform the following minimization.
$$A^{\star} = \argmin_{A' \in \mathcal{A}(A)} \sum_{i = 1}^m H_0(A_i) = \underset{A' \in \mathcal{A}(A)}{\argmin} \sum_{i = 1}^m \sum_{v \in \mathrm{supp} A_i} \frac{1}{\#(v)} \log_2\frac{N}{\#(v)}$$
\end{definition}

We will refer to the problem given in Definition~\ref{def:permute_problem} as the \textit{Row Permutation Problem}. The (standard) notation $\pi_i(j)$ refers to a permutation on $\{m\}$, or a bijection from the set of numbers $\{1,...m\}$ onto itself. Note that the problem allows for different permutations of items in each row.

\textbf{Greedy Method to Minimize Entropy:} We propose a greedy heuristic method to optimize the column-wise entropy by iteratively swapping entries in the rows of $A$. The core of the method is to identify the value $v$ that can be relocated to group as many identical values into the same column as possible. This requires knowledge of two things: the rows where $v$ is present and eligible for relocation, and the possible destination columns where $v$ could be placed in each row. Once the optimal value and destination column have been identified, we swap the corresponding columns of $A$ and mark the value and column as ineligible for relocation in the future.

\begin{table*}[t]
  \centering
  \addtolength{\tabcolsep}{-4pt} 
  \begin{tabular}{ l|c|c|c|c|c|c|c} 
\toprule
Dataset &
($n$,$m$) &
$H_0$ (bits) &
Flat &
HT &
MPH \cite{indeedMPH} &
Succinct \cite{agarwal2015succinct} &
CARAMEL 
\\
\midrule
Synthetic (Uniform Distribution) & 100K, 1K & 9959 & 400 MB & 402 MB & 400.7 MB & 360 MB & 147 MB (2.7x) \\ 
Synthetic (Uniform Distribution)  & 10K, 1K & 9892 & 40 MB & 40.13 MB & 40.07 MB &  34.8 MB &  22.8 MB (1.8x)\\ 
Synthetic (Power Law) & 100K, 1K & 2343 & 400 MB & 402 MB & 400.7 MB &  207 MB &  38.24 MB (10.5x) \\ 
Synthetic (Power Law) & 10k, 1K & 2325 & 40 MB & 40.13 MB & 40.07 MB &  19.5 MB &  5.56 MB (7.2x)\\
\midrule
Featurized Genomes ~\cite{o2016reference} & 125k, 1k & 8192 & 1.0 GB & 1.03 GB & - &  736 MB & 324.29 MB (3.1x) \\ 
Featurized Proteomes ~\cite{o2016reference} & 125k, 1k & 8192 & 1.0 GB & 1.03 GB  & - &  736 MB &  324.29 MB (3.1x) \\ 
\midrule
Count Sketch (AOL terms~\cite{pass2006picture}) & $2^{18}$, 5 & 7.9 & 5.24 MB & 5.6 MB & 5.25 MB &  5.90 MB &  1.56 MB (3.9x) \\ 
Count Sketch (URL~\cite{kaggleURL,mamun2016detecting}) & 40k, 20 & 4.7 & 3.20 MB & 3.28 MB & 3.31 MB &  1.55 MB &  0.55 MB (5.9x) \\ 
\midrule
Amazon Polarity~\cite{mcauley2013hidden,zhang2015character} & 3.6M, 128 & 904.2 & 1.84 GB & 1.92 GB & 1.88 GB &  1.15 GB &  484 MB (3.8x) \\ 
MSMarco~\cite{nguyen2016ms} & 3.2M, 128 & 1432 &  1.65 GB  & 1.714 GB & 1.682 GB & 1.16 GB &  671 MB (2.6x) \\ 
MSMarco~\cite{nguyen2016ms} & 3.2M, 512 & 5195  & 6.58 GB & 6.64 GB & 6.62 GB &  - &  2.45 GB (2.6x) \\ 
MSMarco~\cite{nguyen2016ms} & 8.8M, 512 & 883.2 & 18.1 GB & 18.27 GB & 18.2 GB &  - & 1.12 GB (16x) \\ 
Pile (sample)~\cite{gao2020pile}& 7M, 128 & 265.2 & 34.8 GB & 34.9 GB  & - &  - &  2.4 GB (14.5x) \\ 
\midrule
ABC Headlines Embeddings~\cite{DVN/SYBGZL_2018} & 1.2M, 20 & 148.3 & 33.2 MB & 33.2 MB & 30.55 MB & 27.8 MB &  24.1 MB (1.4x) \\ 
Word2vec Embeddings~\cite{mikolov2013efficient} & 4M, 20 & 149.4 & 124.8 MB & 124.8 MB & 102.17 MB & 180 MB &  83.7 MB (1.5x) \\ 
SIFT Embeddings~\cite{aumuller2017ann} & 1M, 32 & 243.3 & 39.9 MB & 39.9 MB  & 36.9 MB &  67.5 MB &  34.1 MB (1.25x) \\ 
Yandex Embeddings~\cite{yandexT2I} & 1B, 20 & 159.3 & 27.7 GB & 27.7 GB  & 30.27 GB &  - &  22.2 GB (1.25x) \\ 
S-BERT Embeddings~\cite{reimers2019sentence} & 2.9M, 96 & 679.5 & 277.1 MB & 277.1 MB & 302.8 MB &  257 MB & 273.4 MB (1.01x)\\ 
\bottomrule
\end{tabular}
\caption{Compression rates for our method (CARAMEL), hash tables (HT), and the original dataset (Flat). Results are grouped by application. We report the space as the total structure size (in MB / GB) and the compression rate in parenthesis. $H_0$ is the sum of column entropies (in bits).}
\label{tab:csf_space}
\end{table*}

A basic version of the algorithm is listed in Algorithm~\ref{alg:permute_greedy}. In practice, we employ a variety of short-circuit evaluations and heuristics to reduce the time needed to identify the best value and column to relocate. We also allow for early termination when we can no longer relocate a block of at least $B$ values to the same column. . 

\section{Experiments}
\label{sec:experiments}

\textbf{Compression on Synthetic Datasets:} As a preliminary exercise to examine the potential for CSFs to compress integer matrices, we conducted two synthetic data experiments where we generated integer matrices first sampled from a discrete uniform distribution $\mathcal{U}(1, 1000)$ and then from a discrete truncated power law distribution $f(x) = c x^{-k}$ where we set $k=2$, set the support over $(1, 1000)$, and select the normalizing constant $c$ accordingly. This synthetic experiment closely follows that of \cite{8416577} with the notable exception that we examine compressing a two-dimensional matrix as opposed to a one-dimensional array. 

\textbf{Experiment Setup:} Preprocessing work was done in python and resulted in some transformed integer matrix from which to build our CARAMEL structure.  All results in Table~\ref{tab:csf_space} were obtained using using the publicly available Sux4j Library \footnote{http://sux4j.di.unimi.it/} in Java. Compression numbers are recorded excluding keys and we report additional baselines as the sum of the column-wise entropy of the matrix. Query latencies shown in Figure~\ref{fig:ablations} were obtained using timing modules built in C++. For our NLP experiments, we tokenize the text using the popular SentencePiece library \footnote{https://github.com/google/sentencepiece} \cite{kudo2018sentencepiece}. For our case study comparing against a production read-only database at Amazon, we build CSFs using our own Python implementation, which, unlike existing CSF libraries, supports generic hashable types such strings. 

For our remaining baseline methods, we conduct all of the evalations in Java. For Succinct, we use the standard Java release by the authors~\cite{agarwal2015succinct}. For MPH tables, we use a library that has been optimized to work in production by Indeed~\cite{indeedMPH}. It should be noted that Succinct can only support files up to 2 GB in size and that Indeed's MPH library can be slow to construct tables for large inputs. For this reason, we were unable to run the baselines for some of our data compression tasks.

\begin{table*}[t]
\begin{tabular}{l|c|c|c|c|c|c|c}
\hline
Dataset                    & Records & Size & Methods  & Indexing & Index & Query & C-Rate \\ \toprule
\multirow{2}{*}{Amazon Large} & \multirow{2}{*}{1.9M} & \multirow{2}{*}{4.1G} & CARAMEL       &   321.07s           &   1.77G    &     2.42ms & 2.31x \\ \cline{4-8} 
                          & & & Production DB & 76.34s                  & 2.5G       & 7.0ms & 1.64x  \\ \hline
\multirow{2}{*}{Amazon Small} & \multirow{2}{*}{1.0M} & \multirow{2}{*}{1.5G} & CARAMEL       &   149.62s         &    0.88G      &  2.10ms     &  1.7x   \\ \cline{4-8} 
                          & & & Production DB & 31.08s        & 0.95 G      &     6.8ms   & 1.58x      \\ \hline
\end{tabular}
\caption{Comparing between CARAMEL with Amazon Standard look-up table DB. We report statistics on the data size, the indexing time, indexing size, query time, and compression rate (C-Rate) of two methods.}
\label{tab:amazon}
\end{table*}

\textbf{Product-Quantized Embedding Tables:}
It is fairly common in enterprise systems to approximate large embedding tables with quantization \cite{jegou2010product,guo2020accelerating,johnson2019billion,kang2020learning,ko2021low}. 
We apply CARAMEL to quantized embedding tables and report our compression numbers compared to the space of these tables with keys included. Additionally, all embedding tables are quantized to ensure a reasonable metric distortion of >80\% KNN-recall@100 \cite{jegou2010product,kang2020learning}.


\textbf{Construction Time:} Construction times are reasonable, especially in the context of machine learning systems that frequently don't update for upwards of a day or more. For 10M keys we construct a single CSF in just over 1 minute, for 100M keys: 10+ minutes, and for 1B keys: less than 2 hours. For billion scale data this construction time is more than reasonable, especially considering CARAMEL and CSF constructions are trivially parallelizable. 

\textbf{Query Time:} We compare query times between our CARAMEL structure and a standard hash table implemented in C++. We look at both median and 99\textsuperscript{th} percentile query time (P99) as P99 is the defining factor for latency in web-scale production systems. We also investigate the query time scaling as $N$ increases to demonstrate the the worst case complexity of hash-tables against the $O(1)$ lookups into CSFs. Timing experiments were done on a machine equipped with Intel(R) Xeon(R) Gold 5220R CPU @ 2.20GHz over 20,000 queries. 

\subsection{Discussion}
\label{sec:discussion}

We observe that CARAMEL achieves strong compression rates ranging from 1.25-16x across a variety of applications, with better rates for large, low-entropy data. We also see that the compression rate tends to improve as we increase the number of rows $N$. From our synthetic dataset benchmarks, we find that our method is particularly effective in representing power law data, which is commonplace in web-scale settings. Finally, and perhaps surprisingly, we observe that CARAMEL achieves a faster lookup latency than an uncompressed hash table, possibly due to improved cache locality. This result suggests that CARAMEL could be a strictly more effective primitive in latency and memory-critical settings involving integer data than traditional key-value stores that do not exploit the underlying entropy.

\subsection{Case Study on Amazon Product Search}
Amazon's product search engine stores highly frequent queries and their corresponding frequently purchased or clicked products for the purposes of caching results and producing search ranking features \cite{yang2022can}. During the online ranking model prediction phase, the search engine accesses these features from a read-only lookup table. Due to very tight latency requirements for this lookup process, these database queries typically must be executed in less than 10 milliseconds \cite{luo2022rose}. 

Amazon uses an in-house read-only database to host this lookup request. This read-only database is a read-only file based map store similar to Berkeley DBs~\cite{olson1999berkeley}. We refer to this lookup table as Amazon DB. Amazon DB is an industry standard solution for read-only key-value lookups used, amongst other applications, for serving precomputed search results produced by machine learning models. We compare CARAMEL with Amazon DB for indexing query-product ranking features. For these experiments, we use the Python CSF library we developed that, unlike existing tools, is capable of operating on generic hashable types such as strings. 

We sampled two datasets of query-product pairs derived from Amazon.com search logs. This data set contains the query and the associated products that were clicked in one day. Table \ref{tab:amazon} shows the statistics of the datasets we evaluate in this case study.

From Table \ref{tab:amazon}, we observe CARAMEL has more efficient memory usage than the Amazon DB. On the larger benchmark Amazon dataset, CARAMEL achieves a 2.31x compression rate compared to the 1.64x mark of Amazon DB with a considerably faster query time. We observe a similar trend on the smaller dataset. We also observe that the space and latency improvements of CARAMEL come at the cost of longer index construction times. However, we note that CARAMEL's construction costs are still on the order of minutes for these datasets which is reasonable given the one-time cost of constructing the index. Moreover, we believe that CARAMEL can achieve faster construction times by better utilizing the multiprocessing capabilities available in Python, an optimization we leave for future work. 





\section{Conclusion}
We introduce CARAMEL, a new data structure for read-only key-value lookups based on the key insight of combining multiple compressed static functions (CSFs) together. We demonstrate in this paper that CSFs are practically useful data structures that provide a powerful primitive for improving the scalability of numerous data-intensive applications. We also show that CARAMEL performs favorably when compared to existing algorithms and systems proposed in both the academic literature and in active production use at a leading online technology company, which underscores the relevance of our work for industry practitioners. We hope that our work will bring the community's attention to both the prevalence of low-entropy key-value mappings across data-intensive applications and the potential for representations like CARAMEL to aid in the important effort of scaling data infrastructure in the future.

\bibliographystyle{plain}
\bibliography{main}


\end{document}